\documentclass[journal,twocolumn]{IEEEtran}
\usepackage{amsfonts}    
\usepackage{color}
\usepackage{graphicx}
\usepackage[dvips]{epsfig}
\usepackage{graphics}
\usepackage{cite}
\usepackage{arydshln}
\usepackage{amsthm}
\usepackage{amsmath} 
\usepackage{amssymb}  
\usepackage[boxed,ruled,lined]{algorithm2e}                             

\usepackage{amsfonts}    
\usepackage{color}
\usepackage{graphicx}
\usepackage[dvips]{epsfig}
\usepackage{graphics}
\usepackage{cite}
\usepackage{arydshln}
\usepackage{amsmath} 
\usepackage{amssymb}  
\usepackage[boxed,ruled,lined]{algorithm2e}

\def\tiltheta{\tilde{\theta}}
\def\myref #1{(\ref{#1})}
\def\bartheta{\bar{\theta}}
\def\tillambda{\bar{\lambda}}
\def\tilmu{\bar{\mu}}

\def\bR{\mathbb{R}}

\def\cG{\mathcal{G}}

\def\bee{\begin{equation}}
\def\ene{\end{equation}}
\def\beq{\begin{eqnarray}}
\def\enq{\end{eqnarray}}

\newtheorem{myass}{Assumption}
\newtheorem{mylem}{Lemma}
\newtheorem{myrem}{Remark}

\newtheorem{mythm}{Theorem}
\newtheorem{myprop}{Proposition}

\def\barx{\bar{x}}

\begin{document}
\title{Distributed Convex Optimization with Inequality Constraints over Time-varying Unbalanced Digraphs}
\author{Pei Xie, Keyou You, Roberto Tempo, Shiji Song and Cheng Wu \thanks{This work was supported by the National Natural Science Foundation of China (61304038),  Tsinghua University Initiative Scientific Research Program, and CNR International Joint Lab COOPS.} \thanks{P. Xie, K. You, S. Song and C. Wu are with the Department of Automation and TNList, Tsinghua University, 100084, China (emails: xie-p13@mails.tsinghua.edu.cn, \{youky, shijis, cwu\}@tsinghua.edu.cn).} \thanks{R. Tempo is with CNR-IEIIT, Politecnico di Torino, Torino, 10129, Italy (email: roberto.tempo@polito.it).}}

\maketitle
\begin{abstract}
This paper considers a distributed convex optimization problem with inequality constraints over {\em time-varying unbalanced} digraphs, where the cost function is a sum of local objectives, and each node of the graph only knows its local objective and inequality constraints. Although there is a vast literature on distributed optimization, most of them require the graph to be balanced, which is quite restrictive and not necessary. Very recently, the unbalanced problem has been resolved only for either time-invariant graphs or unconstrained optimization. This work addresses the unbalancedness by focusing on an epigraph form of the constrained optimization. A striking feature is that this novel idea can be easily used to study time-varying unbalanced digraphs. Under local communications, a simple iterative algorithm is then designed for each node. We prove that if the graph is uniformly jointly strongly connected, each node asymptotically converges to some common optimal solution.
\end{abstract}

\section{INTRODUCTION}
Over the last decades, a paradigm shift from centralized processing to highly distributed systems has excited interest due to the increasing development in interactions between computers, microprocessors and sensors. This work considers a distributed constrained optimization problem over graphs, where each node only accesses its local objective and constraints.  It arises in network congestion problems, where routers individually optimize their flow rates to minimize the latency along their routes in a distributed way. Other applications include non-autonomous power control, resource allocation, cognitive networks, statistical inference, and machine learning. Without a central coordination unit, each node is unable to obtain the overall information of the optimization problem.   Different from the vast body of literature, we focus on time-varying unbalanced digraphs for constrained optimization problems.

Distributed optimization over graphs has been extensively studied in recent years, see \cite{boyd2011distributed,nedic2015distributed,xi2015distributed} and references therein.  In the seminal work, the authors of \cite{nedic2009distributed} propose a distributed gradient descend (DGD) algorithm by designing a communication protocol to achieve {\em consensus} (or reach an agreement) among nodes, and using (sub)-gradients to individually  minimize each local objective. However, this algorithm relies on {\em balanced} graphs via a sequence of doubly-stochastic matrices, which is critical to the success of the algorithm.  It surely limits the applicability of the algorithm in practice, especially for time-varying digraphs, as the balanced condition is not easy to check in a distributed way. For unbalanced digraphs, the DGD algorithm actually minimizes a {\em weighted} sum of local objectives, rather than the sum in the optimization.

To resolve it, the authors of \cite{xi2015distributed}  adopt the ``surplus-based'' idea from \cite{cai2012average} by augmenting an auxiliary vector for each node to record the state updates, the goal of which is to achieve {\em average consensus} over time-invariant digraphs.  Although this idea has been extended to time-varying digraphs  in \cite{xu2015fully}, it only addresses the resource allocation problem. In \cite{moral2016Perronestimate},  an additional distributed algorithm is designed to compensate unbalancedness. As explicitly pointed out by the author, it is not directly applicable to time-varying digraphs. It is worthy mentioning that the distributed algorithm of alternating directed method of multipliers  \cite{iutzeler2016explicit} even requires the graph to be undirected.

For time-varying unbalanced digraphs, the authors of \cite{nedic2015distributed} propose an algorithm by combining the gradient descent and the push-sum mechanism, which is primarily designed to achieve average-consensus on unbalanced digraphs.  Yet, this algorithm only works for {\em unconstrained} optimization as push-sum based updates do not preserve feasibility of the iterates. Moreover, it is complicated and involves multiple nonlinear iterations, which also appears in the algorithm of \cite{sun2016distributed}. More importantly, each node requires to solve an optimization problem per iteration in \cite{sun2016distributed}.

Overall, the problem is largely open for the distributed constrained optimization over time-varying unbalanced digraphs. To solve these issues simultaneously, we introduce an epigraph form of the constrained optimization problem to convert the objective function to a linear form, which is in sharp contrast with the existing approaches to resolve unbalancedness. The advantage of the epigraph form is very clear as it eliminates the differences among local objective functions, and thus can easily handle unbalancedness. To design a novel two-step recursive algorithm, we firstly solve an unconstrained optimization problem without the decoupled constraints in the epigraph form by using a standard distributed sub-gradient descent and obtain an intermediate state vector in each node. Secondly, the intermediate state vector of each node is moved toward the intersection of its local constraint sets by using the distributed version of Polyak random algorithm, which is proposed in \cite{you2016scenario}. We further introduce an additional ``projection'' toward a fixed direction to improve the transient performance. More importantly, this paper significantly extends the algorithm in \cite{you2016scenario} to the case of time-varying graphs. Our algorithm converges to an optimal solution if the time-varying unbalanced digraphs are uniformly jointly strongly connected.

The rest of this paper is organized as follows. In Section \ref{Formulation}, we formulate the distributed constrained optimization and review the existing works to motivate our ideas. In Section \ref{Algorithm}, the epigraph form of the original optimization is introduced to attack the unbalanced issue, and a random-fixed projected algorithm is then designed to distributedly solve the reformulated optimization. In Section \ref{Convergence}, the convergence of the proposed algorithm is rigorously proved. Finally, some concluding remarks are drawn in Section \ref{Conclusion}.

\textbf{Notation}: For two vectors $a = \left[a_1,...,a_n\right]^T$ and $b = \left[b_1,...,b_n\right]^T$, the notation $a\preceq b$ means that $a_i\leq b_i$ for any $i\in\{1,...,n\}$. Similar notation is used for $\prec$, $\succeq$ and $\succ$. The symbols $1_n$ and $0_n$ denotes the vectors with all entries equal to one and zero respectively, and $e_j$ denotes a unit vector with the $j$th element equals to one. For a matrix $A$, we use $\|A\|$ and $\rho(A)$ to represent its reduced norm and spectral radius respectively. $\otimes$ denotes the Kronecker product. The subdifferential of  $f$ at $x\in\mathbb{R}^m$ is denoted by $\partial f(x)$. Finally, $f(\theta)_+=\max\{0,f(\theta)\}$ takes the nonnegative part of $f$.

\section{Problem Formulation and Motivation}
\label{Formulation}
\subsection{Distributed Contrained Optimization}
This work considers a network of $n$ nodes to collaboratively solve a convex optimization with inequality constraints,
\begin{align}
\min_{x\in X}\quad&F(x)\triangleq\sum_{i=1}^n f_i(x),\notag\\
\label{equ:p0def}
\text{subject to}\quad&g_i(x)\preceq0,\text{ }i=1,2,...,n,
\end{align}
where $ X\in\mathbb{R}^m$ is a common convex set to all nodes, while $f_i: \mathbb{R}^m\rightarrow\mathbb{R}$ is a convex function only known by node $i$. Moreover, only node $i$ is aware of its local constraints $g_i(x)\preceq0$, where $g_i(x)=\left[g_i^{1}(x),...,g_i^{\tau_i}(x)\right]^T\in\mathbb{R}^{\tau_i}$ is a vector of convex functions.

A digraph $\mathcal{G}=\left(\mathcal{V},\mathcal{E}\right)$ is introduced to describe interactions between nodes, where $\mathcal{V}:=\{1,...,n\}$ denotes the set of nodes, and the set of interaction links is represented by $\mathcal{E}$. A directed edge $(i,j)\in\mathcal{E}$ if node $i$ can directly receive information from node $j$. We define $\mathcal{N}_i^{in}=\{j|(i,j)\in\mathcal{E}\}$ as the collection of in-neighbors of node $i$, i.e., the set of nodes directly send information to node $i$. The out-neighbors $\mathcal{N}_i^{out}=\{j|(j,i)\in\mathcal{E}\}$ are defined similarly. Each node is included in its out-neighbors and in-neighbors. Node $i$ is \textit{reachable} by node $j$ if there exist $i_1,...,i_p\in\mathcal{V}$ such that $(i,i_1),...,(i_{k-1},i_k),...,(i_p,j)\in\mathcal{E}$.  $\mathcal{G}$ is \textit{strongly connected} if every node is reachable by all nodes. A weighting matrix adapted to $\mathcal{G}$ is defined as $A$ = $\{a_{ij}\}\in\mathbb{R}^{n\times n}$, which satisfies that $a_{ij}>0$ if $(i,j)\in\mathcal{E}$ and $a_{ij}=0$, otherwise.  $\mathcal{G}$ is said to be \textit{balanced} if $\sum_{j\in\mathcal{N}_i^{out}}a_{ji}=\sum_{j\in\mathcal{N}_i^{in}}a_{ij}$ for any $i\in\mathcal{V}$, and \textit{unbalanced}, otherwise. Moreover, $A$ is \textit{row-stochastic} if $\sum_{j=1}^na_{ij}=1$ for any $i\in\mathcal{V}$, \textit{column-stochastic}, if $\sum_{i=1}^na_{ij}=1$ for any $j\in\mathcal{V}$, and \textit{doubly-stochastic} if it is both row-stochastic and column-stochastic. If the links among nodes change with time, we use $\{\mathcal{G}^k\}=\{\left(\mathcal{V},\mathcal{E}^k\right)\}$ and $\{A^k\}$ to represent time-varying graphs and the associated weighting matrices.

The objective of this work is to design a recursive algorithm to distributedly solve problem \myref{equ:p0def} over time-varying unbalanced digraphs. That is, each node $i$ communicates with its neighbors and locally updates a vector $x_i^k$ so that each $x_i^k$ eventually converges to some common optimal solution.

\subsection{Distributed Gradient Descend Algorithms}
In the standard DGD algorithm \cite{nedic2009distributed}, each node $i$ updates its local estimate of an optimal solution by
\begin{equation}
\label{equ:dgd}
x_i^{k+1}=\sum_{j=1}^na_{ij}x_j^k-\zeta^k\nabla f_i,
\end{equation}
where $\zeta^k$ is a given step size.

However, the DGD is only able to solve the optimization problem over balanced digraphs, which is not applicable to unbalanced graphs. To illustrate this point, we utilize the \textit{Perron vector} \cite{horn2012matrix} of a weighting matrix $A$ as follows.

\begin{mylem}[Perron vector]\label{lemma:perron} If $\mathcal{G}$ is a strongly-connected digraph and $A$ is the associated row-stochastic weighting matrix, there exists a Perron vector $\pi\in\mathbb{R}^n$ such that
\begin{equation}
\label{equ:l1res1}
\pi^TA=\pi^T, \pi^T1_n=1,~\text{and}~\pi_i>0.
\end{equation}
\end{mylem}
By multiplying $\pi_i$ in \myref{equ:l1res1} on both sides of \myref{equ:dgd} and summing up over $i$, we obtain that
\begin{align}
\label{equ:sdgd}
  \barx^{k+1}&\triangleq\sum_{i=1}^n\pi_ix_i^{k+1} \notag\\
             &=\sum_{j=1}^n\big(\sum_{i=1}^n\pi_i a_{ij}\big)x_j^k-\zeta^k\sum_{i=1}^n\pi_i\nabla f_i(x_i^k)\notag\\
             &=\barx^{k}-\zeta^k\sum_{i=1}^n\pi_i\nabla f_i(x_i^k).
\end{align}
If all nodes have reached consensus, then \myref{equ:sdgd} is written as
\begin{equation}
\label{equ:cgd}
\barx^{k+1}=\barx^k-\zeta^k\sum_{i=1}^n\pi_i\nabla f_i(\barx^k).
\end{equation}
Clearly, \myref{equ:cgd} is a DGD algorithm to minimize a new function
\begin{equation}
\label{barf}
\bar{F}(x)\triangleq\sum_{i=1}^n\pi_if_i(x),
\end{equation}
and each node converges to a minimizer of $\bar{F}(x)$ rather than $F(x)$ in \myref{equ:p0def}, which is also noted in \cite{xi2015distributed}. For a generic unbalanced digraph, the weighting matrix is no longer doubly-stochastic, and the Perron vector is not equal to $\left[\frac{1}{n},...,\frac{1}{n}\right]^T$,  which obviously implies that $\bar{F}(x)\neq F(x)$. That is, DGD in \myref{equ:dgd} can not be applied to the case of unbalanced graphs.

If each node $i$ is able to access its associated element of the Perron vector $\pi_i$, it follows from \myref{equ:cgd} that a natural way to modify the DGD in \myref{equ:dgd} is given as
$$
x_i^{k+1}=\sum_{j=1}^na_{ij}x_j^k-\frac{\zeta^k}{\pi_i}\nabla f_i^k,
$$
which is recently exploited in \cite{moral2016Perronestimate} by designing an additional distributed algorithm to locally estimate $\pi_i$. However, it is not directly applicable to time-varying graphs as {\em constant} Perron vector does not exist. In fact, this shortcoming has also been explicitly pointed out in \cite{moral2016Perronestimate}. Another idea to resolve the unbalanced problem is to augment the original row-stochastic matrix into a doubly-stochastic matrix. This novel approach is originally proposed by \cite{cai2012average} for average consensus problems over unbalanced graphs. The key idea is to augment an additional variable for each agent, called ``surplus",  which is used to locally record individual state updates. In \cite{xi2015distributed},  the ``surplus-based'' idea is adopted to solve the distributed optimization problem over fixed unbalanced graphs. Again, it is unclear how to use the ``surplus-based'' idea to  solve the distributed optimization problem over time-varying unbalanced digraphs. This problem has been resolved in \cite{nedic2015distributed}  by adopting the so-called push-sum consensus protocol, the goal of which is to achieve the average consensus over unbalanced graphs. Unfortunately, the algorithms appear to be over complicated and involve nonlinear iterations.  More importantly,  they are restricted to the unconstrained optimization, and the algorithm is not as intuitive as the DGD.

This work solves the unbalanced problem from a different perspective, which can easily address the constrained optimization over time-varying unbalanced digraphs.
\section{Distributed Algorithms Over Time-varying Graphs for Constrained Optimization}
\label{Algorithm}
As noted, perhaps it is not effective to attack the unbalanced problem via the Perron vector. To overcome this limitation, we study the epigraph form of the optimization \myref{equ:p0def}, and obtain the same linear objective function in every node. This eliminates the effect of different elements of the Perron vector on the limiting point of \myref{equ:cgd}. Then we utilize the DGD in \myref{equ:dgd} to resolve the epigraph form and obtain an intermediate state vector. The feasibility of the local estimate in each node is asymptotically guaranteed by further driving this vector toward the constraint set. That is,  we update  the intermediate vector toward the negative sub-gradient direction of a local constraint function.  This novel idea is in fact proposed in the recent work \cite{you2016scenario}, which however only focuses on time-invariant digraphs. In this work, we extend it to time-varying case.
\subsection{Epigraph form of the constrained optimization}
Our main idea does not focus on $\pi_i$ but on $f_i$ in \myref{equ:cgd}. Specifically, if we transform all the local objective $f_i(x)$ to the same form $f_0(x)$, then \myref{equ:cgd} is reduced to $\barx^{k+1}=\barx^k-\zeta^k\nabla f_0(\barx^k)$, which implies that there is no difference balanced and unbalanced digraphs. We shall exploit this idea to design our distributed algorithm. 

Given $f(x):\mathbb{R}^m\rightarrow\mathbb{R}$, the epigraph of $f$ is defined as
\begin{equation*}
\text{epi} f =\{(x,t)|x\in \text{dom} f, f(x)\leq t\},
\end{equation*}
which is a subset of $\mathbb{R}^{m+1}$. It follows from \cite{bertsekas2015convex} that the epigraph of $f$ is a convex set if and only if $f$ is convex, and minimizing $f$ is equal to searching the minimal auxiliary variable $t$ within the epigraph. This  transforms the optimization problem of minimizing a convex objective to minimizing a \textit{linear function} within a convex set. In the case of multiple functions, it can be defined similarly via multiple auxiliary variables.

Now, we consider the epigraph form of \myref{equ:p0def} by using an auxiliary vector $t\in\mathbb{R}^n$. It is clear that problem \myref{equ:p0def} can be reformulated as
\begin{align}
\min_{(x,t)\in\Theta}\quad&\sum_{i=1}^n 1_n^Tt/n,\notag\\
\text{subject to}\quad&f_i(x)-e_i^Tt\leq 0,\notag\\
\label{equ:p1def}
&g_i(x)\preceq0,\quad i=1,2,...,n,
\end{align}
where $\Theta= X\times\mathbb{R}^n$ is the Cartesian product of $ X$ and $\mathbb{R}^n$. Clearly, local objectives are transformed to the same linear function $f_0(x,t)=1_n^Tt/n$.
\begin{myrem}\label{rem:epi} In view of the epigraph form, we have the following comments.
\begin{enumerate}\renewcommand{\labelenumi}{\rm(\alph{enumi})}
\item Denote $y=[x^T,t^T]^T$ and $f_0(y)=c^Ty/n$, where $c=[0_m^T,1_n^T]^T$. Thus, the objective in \myref{equ:p1def} becomes the sum of the local objective $f_0$, which is the same for all nodes. Therefore, $\bar{F}(y)$ in \myref{barf} is reduced to $\bar{F}(y)=nf_0(y)=F(y)$. In this way, we lift the limitations of DGD and make it effective for unbalanced graphs.
\item The local objective $f_i(x)$ in \myref{equ:p0def} is handled via an additional constraint in \myref{equ:p1def} such that $\tilde{f}_i(y)= f_i(x)-e_i^Tt\leq0$, where $\tilde{f}_i(y)$ is a convex function as well. To evaluate $\tilde{f}_i(y)$, it requires each node $i$ to select the $i$-th element of the vector $t$. As $i$ is the identifier of node $i$, the epigraph form requires each node to know its identifier, which is also needed in \cite[Assumption 2]{moral2016Perronestimate}.
\end{enumerate}\renewcommand{\labelenumi}{\rm(\alph{enumi})}
\end{myrem}
\subsection{Distributed random-fixed projected algorithm}
To recursively solve \myref{equ:p1def}, every node $j$ maintains a local estimate $x_j^k\in\mathbb{R}^m$ and $t_j^k\in\mathbb{R}^n$ at each iteration $k$.  Each node $i$ updates $x_i,t_i$ by combining its in-neighbors' estimates and adjusting it to approach its local constraint set. 

Specifically, we first solve an unconstrained optimization problem which removes the constraints in \myref{equ:p1def} by using the standard DGD algorithm and obtain intermediate state vectors $p_j^k$ and $y_j^k$, which correspond to $t_j^k\in\mathbb{R}^m$ and $x_j^k\in\mathbb{R}^n$, respectively, i.e.,
\begin{align}
\label{equ:a1p1}
p_j^k&=\sum_{i=1}^na_{ji}^kt_i^k-\zeta^k1_n,\\
\label{equ:a1p2}
y_j^k&=\sum_{i=1}^na_{ji}^kx_i^k,
\end{align}
where $\zeta^k$ is the step-size satisfying the persistent exciting condition
\begin{equation}\label{equ:stepdef}
\zeta^k>0,\quad\sum_{k=0}^\infty\zeta^k=\infty,\quad\sum_{k=0}^\infty(\zeta^k)^2<\infty.
\end{equation}

Then, we adopt the Polyak's idea to address the constraints of \myref{equ:p1def} to drive the intermediate state vectors toward the feasible set. Define the following notations
\begin{align}
X_j^l&=\{x\in \mathbb{R}^m|g_j^l(x)\leq0\}, l\in\{1,...,\tau_j\},\notag\\
\label{equ:oldsetdef}
X\times T_j&=\{(x,t)|x\in X,f_j(x)-e_j^Tt\leq0\}.
\end{align}
We update $y_j^k$ toward a randomly selected set $X_j^{\omega_j^k}$ by using the Polyak's projection idea, i.e.,
\begin{equation}
\label{equ:a1p3}
z_j^{k}=y_j^k-\beta\frac{g_j^{\omega_j^k}(y_j^k)_+}{\|u_j^k\|^2}u_j^k,
\end{equation}
Where $\beta\in(0,2)$ is a constant parameter, and the vector $u_j^k\in\partial g_j^{\omega_j^k}(y_j^k)_+$ if $g_j^{\omega_j^k}(y_j^k)_+>0$ and $u_j^k=u_j$ for some $u_j\neq0$ if $g_j^{\omega_j^k}(y_j^k)_+=0$. In fact, $u_j^k$ is a decreasing direction of $g_j^{\omega_j^k}(y_j^k)_+$, i.e.,
$d(z_j^k,  X_j^{\omega_j^k})\leq d(y_j^k,  X_j^{\omega_j^k})$ for sufficiently small $\beta$ where $d(x,X)$ denotes the distance from $x$ to the closed convex set $X$. If $\omega_j^k$ is appropriately selected, it is expected in average that
$d(z_j^k, \cap_{l=1}^{\tau_j} X_j^l)\leq d(y_j^k, \cap_{l=1}^{\tau_j} X_j^l).$

While the auxiliary vector $t_j^k$ is not updated during the above process, we use the same idea to handle the newly introduced constraint $X_j^0\times T_j$ such that
\begin{align}
\label{equ:a1p4}
x_j^{k+1}&=\Pi_X(z_j^k-\beta\frac{(f_j(z_j^k)-e_j^Tp_j^k)_+}{1+\|v_j^k\|^2}v_j^k),\\
\label{equ:a1p5}
t_j^{k+1}&=p_j^k+\beta\frac{(f_j(z_j^k)-e_j^Tp_j^k)_+}{1+\|v_j^k\|^2}e_j,
\end{align}
where $v_j^k\in\partial f(z_j^k)$. Similarly, we have that
\begin{equation}
d((x_j^{k+1},t_j^{k+1}), X\times T_j)\leq d((z_j^k,t_j^k), X\times T_j).
\end{equation}
Once all the nodes reach an agreement, the state vector $(x_j^k,t_j^k)$ in each node asymptotically converges to a feasible point.
\begin{algorithm}[htbp!]
\caption{Distributed random-fixed projected algorithm (D-RFP)}
\begin{enumerate}\renewcommand{\labelenumi}{\noindent\arabic{enumi}:}
\item \textbf{Initialization:} For each node $j\in\mathcal{V}$, set $x_j=0,t_j=0$.
\item\textbf{Repeat}
\item\textbf{Set} $k=1$.
\item\textbf{Local information exchange:} Each node $j\in\mathcal{V}$ broadcasts $x_j$ and $t_j$ to its out-neighbors.
\item\textbf{Local variables update:} Every node $j\in\mathcal{V}$ receives the state vectors $x_i$ and $t_i$ from its in-neighbors $i\in\mathcal{N}_j^{in}$ and updates its local vectors as follows
\begin{itemize}
  \item $y_j=\sum_{i\in\mathcal{N}_j^{in}}a_{ji}^kx_i$, $p_j=\sum_{i\in\mathcal{N}_j^{in}}a_{ji}^kt_i-\zeta^k1_n$, where the stepsize $\zeta^k$ is given in \myref{equ:stepdef}.
  \item Draw a random variable $\omega_j$ from $\{1,...,\tau_j\}$, and obtain $z_j=y_j-\beta\frac{g_j^{\omega_j}(y_j)_+}{\|u_j\|^2}u_j$, where $u_j$ is defined in \myref{equ:a1p3}.
  \item Set $x_j\leftarrow\Pi_X(z_j-\beta\frac{(f_j(z_j)-e_j^Tp_j)_+}{1+\|v_j\|^2}v_j)$, where $v_j$ is defined in \myref{equ:a1p4}, and $t_j\leftarrow p_j+\beta\frac{(f_j(z_j)-e_j^Tp_j)_+}{1+\|v_j\|^2}e_j$.
\end{itemize}
\item\textbf{Set} $k=k+1$.
\item\textbf{Until} a predefined stopping rule is satisfied.
\end{enumerate}
\label{alg_disrandom}
\end{algorithm}

\begin{myrem}\label{rem:simple} In light of the constraints in \myref{equ:p1def}, we fully exploit their structures by adopting the Polyak projection idea \cite{nedic2011random}, which is different from the constrained version of DGD in \cite{lobel2011distributed} by directly using the projection operator. Clearly, the projection is easy to perform only if the projected set has a relatively simple structure, e.g., interval or half-space. From this perspective, our algorithm needs not to accurately calculate the projection and thus requires much less computational load per iteration. 
\end{myrem}
\begin{myrem}\label{rem:polyak1} Algorithm \ref{alg_disrandom}  is motivated by a centralized Polyak random algorithm \cite{nedic2011random}, which is extended to the distributed version in \cite{you2016scenario}. The main difference from \cite{you2016scenario} is that we do not use randomized projection on all the constraints. For instance, $X_j^0\times T_j$ is always considered per iteration. If we equally treat the constraints $g_j(x)\preceq 0$ and $f_j(x)-e_j^Tt\leq 0$, then once the selected constraint is from an element of $g_j(x)$, the vector $t$ is not updated as $t$ is independent of $g_j(x)$. This may slow down the convergence speed and lead to undesired transient behavior. Thus, Algorithm \ref{alg_disrandom}  adds a fixed projection to ensure that both $x$ and $t$ are updated at each iteration.
\end{myrem}
\begin{myrem}\label{rem:polyak2} Algorithm \ref{alg_disrandom} also has a strong relation with the alternating projection algorithm, which finds the intersection of several constraint sets by alternative projections, see e.g. \cite{escalante2011alternating}. The idea is that the state vector asymptotically gets closer to the intersection by repeatedly projecting to differently selected constraint sets. In light of this, the ``projection'' in Algorithm \ref{alg_disrandom}  can be performed for an arbitrary number of times at a single iteration as well. Another option is to project to the constraint set with the farthest distance from the intermediate vector.
\end{myrem}
Under some mild conditions,  we prove the convergence of Algorithm \ref{alg_disrandom} in next section.  

\section{Convergence Analysis}
\label{Convergence}
\subsection{Assumptions and Convergence}
To simplify notations in proving the convergence of Algorithm \ref{alg_disrandom}, we consider the following form of \myref{equ:p1def}
\begin{align}
\min_{\theta\in\Theta}\quad&c^T\theta,\notag\\
\text{s.t.}\quad &f_j(\theta)\leq0,\notag\\
\label{equ:p2def}
&g_j(\theta)\preceq0,\quad j=1,2,...,n,
\end{align}
where $\theta=(x,t) \in\bR^d$ in  \myref{equ:p1def} and $d=m+n$.  Moreover, $f_j:\mathbb{R}^d\rightarrow\mathbb{R}$ is a convex function and $g_j:\mathbb{R}^d\rightarrow\mathbb{R}^{\tau_j}$ is a vector of convex functions. Then, Algorithm \ref{alg_disrandom} for \myref{equ:p2def} is given as
\begin{subequations}
\label{equ:generic}
\begin{align}
\label{equ:a2p1}
p_j^k&=\sum_{i=1}^na_{ji}^k\theta_i^k-\zeta^kc,\\
\label{equ:a2p2}
q_j^k&=p_j^k-\beta\frac{g^{\omega_j^k}(p_j^k)_+}{\|u_j^k\|^2}u_j^k,\\
\label{equ:a2p3}
\theta_j^{k+1}&=\Pi_\Theta(q_j^k-\beta\frac{f_j(p_j^k)_+}{\|v_j^k\|^2}v_j^k),
\end{align}
\end{subequations}
where $u_j^k\in\partial g_j^{\omega_j^k}(p_j^k)_+$ if $g_j^{\omega_j^k}(p_j^k)_+>0$ and $u_j^k=u_j$ for some $u_j\neq0$ if $g_j^{\omega_j^k}(p_j^k)_+=0$, and the vector $v_j^k$ is defined similarly related to $f_j$. Clearly, it is sufficient to prove the convergence of  \myref{equ:generic}. To this end, we introduce some notations
\begin{align}
\Theta_j&=\{\theta\in\Theta|f_j(\theta)\leq0,g_j(\theta)\preceq0\},\notag\\
\Theta_0&=\Theta_1\cap\cdots\cap\Theta_n,\notag\\
\label{def:theta}
\Theta^*&=\{\theta\in\Theta_0|c^T\theta\leq c^T\theta',\forall\theta'\in\Theta_0\}
\end{align}
and make the following assumptions. 

\begin{myass}[Randomness and bounded subgradients]\label{assum:random} \quad
\begin{enumerate}\renewcommand{\labelenumi}{\rm(\alph{enumi})}
\item $\{\omega_j^k\}$ is an i.i.d. sequence, and is uniformly distributed over $\{1,...,\tau_j\}$ for any $j\in\mathcal{V}$. Moreover, it is independent over the index $j$.
\item The sub-gradient $u_j^k$ and $v_j^k$ given in \myref{equ:generic} are uniformly bounded over the set $\Theta$, i.e., there exists a scalar $D>0$ such that
    \begin{equation*}
    \max\{\|u_j^k\|,\|v_j^k\|\}\leq D,\forall j\in\mathcal{V},\forall k>0.
    \end{equation*}
\end{enumerate}
\end{myass}

\begin{myass}[Centralized solvability]\label{assum:solvable}The optimization problem in \myref{equ:p2def} is feasible and has a nonempty set of optimal solutions, i.e., $\Theta_0\neq\varnothing$ and $\Theta^*\neq\varnothing$.
\end{myass}

\begin{myass}[Uniformly Jointly Strong Connectivity]\label{assum:connected}The time-varying graphs $\{\mathcal{G}^k\}$ are uniformly jointly strong connected, i.e., there exists a positive integer B such that the joint graph $\mathcal{G}^t\cup\mathcal{G}^{t+1}\cup\cdots\cup\mathcal{G}^{t+B-1}$ is strongly connected for any $t\geq0$.
\end{myass}

\begin{myass}[Interaction Strength]\label{assum:strength}There exists a scalar $\gamma>0$ such that for any $i,j\in\mathcal{V}$ and $k\geq0$, if $a_{ij}^k>0$, then $a_{ij}^k\geq\gamma$.
\end{myass}
Assumption \ref{assum:connected} is common, see e.g. \cite{nedic2015distributed}, which
indicates that node $i$ can always directly or indirectly receive information from any other node $j$ during any period of $B$. Assumption \ref{assum:strength} implies that if node $i$ has access to node $j$ (node $i$ also has access to itself), then the strength would not be too weak. 

Now, we are ready to state our main convergence result. 
\begin{mythm}[Almost sure convergence]\label{thm:convergent} Under Assumptions \ref{assum:random}-\ref{assum:strength}, the sequence $\{\theta_j^k\}$ in \myref{equ:generic} almost surely converges to some common point in the set $\Theta^*$ of the optimal solutions to \myref{equ:p2def}.
\end{mythm}
\subsection{Proof of Theorem \ref{thm:convergent}}

The proof  is roughly divided into three parts. The first part establishes a sufficient condition to ensure asymptotic \textit{consensus}, under which the sequence $\{\theta_j^k\}$ converges to the same value for all $j\in\mathcal{V}$. The second part demonstrates the asymptotic \textit{feasibility} of the state vector $\theta_j^k$. Finally, the last part guarantees \textit{optimality} by showing that the distance of $\theta_j^k$ to any optimal point $\theta^*$ is ``stochastically'' decreasing. Combining these three parts, we show that $\{\theta_j^k\}$ converges to some common point in $\Theta^*$ almost surely.

Firstly, we establish a sufficient condition for asymptotic consensus, which relies on the following lemmas.
\begin{mylem}\label{thm:ergodic} (\cite{tsitsiklis1984problems}) Let Assumptions \ref{assum:connected}-\ref{assum:strength} hold for $\{\cG^k\}$ and $\{A^k\}$. For $s\geq k$, let $A^{s:k}=A^{s-1}\cdots A^k$ where $A^{k:k}=I$ and $a_{ij}^{s:k}$ be entries of $A^{s:k}$. Then for any $k\geq0$, there exists a normalized vector $\pi^k$ ($1_n^T\pi^k=1$) such that
\begin{enumerate}\renewcommand{\labelenumi}{\rm(\alph{enumi})}
\item We can find $L>0$, $0<\xi<1$ such that $|a^{s:k}_{ij}-\pi^k_j|\leq L\xi^{s-k}$ for any $i,j\in\mathcal{V}$ and $s\geq k$.
\item There exists a constant $\eta\geq\gamma^{(n-1)B}$ such that $\pi^k_i\geq\eta$ for any $i\in\mathcal{V}$ and $k\geq0$.
\item ${(\pi^k)}^T={(\pi^{k+1})}^TA^k$.
\end{enumerate}
\end{mylem}

\begin{mylem}[Exponential stability]\label{lemma:decay}Consider a vector sequence $\{\delta^k\}\in\mathbb{R}^n$ generated by
\begin{align}
\delta^{k+1}=M^k\delta^k+\varepsilon^k,
\end{align}
where $\{M^k\}\in\mathbb{R}^{n\times n}$ is a sequence of matrices. If there exist constant scalars $L>0$ and $0\leq\xi<1$ such that $\|\prod_{l=k}^{s-1}M^l\|\leq L\xi^{s-k}$, and $\lim_{k\rightarrow\infty}\|\varepsilon^k\|=0$, then $\lim_{k\rightarrow\infty}\|\delta^k\|=0$.
\end{mylem}
\begin{proof}
We denote $M^{s:k}\triangleq \prod_{l=k}^{s-1}M^l$ for any $s\geq k\geq0$, then $\|M^{s:k}\|\leq L\xi^{s-k}$. It follows that
$$
\delta^k=M^{k:0}\delta^0+\sum_{l=1}^{k}M^{k:l}\varepsilon^{l-1}.
$$

By the triangle inequality, we have that
\begin{align}
\label{lemma:decayp1}
\|\delta^k\|&\leq\|M^{k:0}\|\|\delta^0\|+\sum_{l=1}^{k}\|M^{k:l}\|\|\varepsilon^{l-1}\|\notag\\
&\leq L\xi^k\|\delta^0\|+\sum_{l=1}^{k}L\xi^{k-l}\|\varepsilon^{l-1}\|.
\end{align}
It is clear that  $L\xi^k\|\delta^0\|\rightarrow0$ as $k\rightarrow\infty$, and the second term is divided into two parts as
\begin{align}
&\sum_{l=1}^{k}L\xi^{k-l}\|\varepsilon^{l-1}\|\notag\\
&=\sum_{l=1}^{k_1}L\xi^{k-l}\|\varepsilon^{l-1}\|+\sum_{l=k_1+1}^{k}L\xi^{k-l}\|\varepsilon^{l-1}\|\notag\\
&=L\xi^{k-k_1}\sum_{l=1}^{k_1}\xi^{k_1-l}\|\varepsilon^{l-1}\|+\sum_{l=k_1+1}^{k}L\xi^{k-l}\|\varepsilon^{l-1}\|\notag\\
\label{lemma:decayp2}
&\leq L\xi^{k-k_1}\frac{\xi}{1-\xi}\sup_{l\ge0}\|\varepsilon^l\|+L\frac{\xi}{1-\xi}\sup_{l\geq k_1}\|\varepsilon^{l}\|.
\end{align}
Since $\lim_{k\rightarrow\infty}\|\varepsilon^k\|=0$, there exists a positive scalar $\hat{\varepsilon}$ such that $\|\varepsilon^k\|\leq \hat{\varepsilon}$, e.g., $\sup_{l\ge0}\|\varepsilon^l\|\le \hat{\varepsilon}$. Then, letting $k\rightarrow\infty$, the first term in (\ref{lemma:decayp2}) tends to zero, after which we let $k_1\rightarrow\infty$, the second term in (\ref{lemma:decayp2}) goes to zero. Thus,
$
\lim_{k\rightarrow\infty}\sum_{l=1}^{k}L\xi^{k-l}\|\varepsilon^{l-1}\|=0.
$
By (\ref{lemma:decayp1}), we obtain that $\lim_{k\rightarrow\infty}\|\delta^k\|=0$.
\end{proof}
\begin{mylem}[Asymptotic consensus]\label{lemma:consensus}Consider the following sequence
\begin{align}
\label{equ:l8def}
\theta_j^{k+1}=\sum_{i=1}^na_{ji}^k\theta_i^k+\epsilon_j^k,\quad\forall j\in\mathcal{V},
\end{align}
where sequence $\{a_{ji}^k\}_{n\times n}$ satisfy Assumptions \ref{assum:connected}-\ref{assum:strength}. Let $\pi^k$ be given in Lemma \ref{thm:ergodic} and $\bartheta^k=\sum_{i=1}\pi_i^k\theta_i^k$. If $\lim_{k\rightarrow\infty}\|\epsilon_j^k\|=0$ for any $j\in\mathcal{V}$, it follows that
\begin{align}
\label{equ:l8res}
\lim_{k\rightarrow\infty}\|\theta_j^k-\bartheta^k\|=0,\quad\forall j\in\mathcal{V}.
\end{align}
\end{mylem}
\begin{proof} Denote $\theta^k={[{(\theta_1^k)}^T,...,{(\theta_n^k)}^T]}^T$ and $\epsilon^k={[{(\epsilon_1^k)}^T,...,{(\epsilon_n^k)}^T]}^T$, then \myref{equ:l8def} can be uniformly written as
\begin{align}
\label{equ:l8p1}
\theta^{k+1}=(A^k\otimes I_d)\theta^k+\epsilon^k.
\end{align}

By pre-multiplying $(1_n\pi^{k+1})\otimes I_d$ on both sides of \myref{equ:l8p1}, it follows from Lemma \ref{thm:ergodic} that
\begin{equation}
\label{equ:l8p2}
(1_n{(\pi^{k+1})}^T\otimes I_d)\theta^{k+1}=(1_n{(\pi^k)}^T\otimes I_d)\theta^k+{\bar{\epsilon}}^k,
\end{equation}
where ${\bar{\epsilon}}^k=(1_n{(\pi^{k+1})}^T\otimes I_d)\epsilon^k$, which obviously converges to zero. Since $\bartheta^k=({(\pi^k)}^T\otimes I_d)\theta^k$, then \myref{equ:l8p2} is rewritten as
\begin{equation}
\label{equ:l8p3}
(1_n\otimes I_d)\bartheta^{k+1}=(1_n\otimes I_d)\bartheta^k+{\bar{\epsilon}}^k.
\end{equation}
Let $\tiltheta^k=\theta^k-(1_n\otimes I_d)\bar{\theta}^k$. Then, our aim is to prove that $\tiltheta^k\rightarrow0$ as $k\rightarrow\infty$. In fact, we subtract \myref{equ:l8p1} by \myref{equ:l8p3} and obtain
\begin{align}
\tiltheta^{k+1}=((A^k-1_n{(\pi^k)}^T)\otimes I_d)\theta^k+\epsilon^k-{\bar{\epsilon}}^k.
\end{align}
We note that $\theta^k=\tiltheta^k+(1_n\otimes I_d)\bartheta^k$, and
\begin{align}
&((A^k-1_n{(\pi^k)}^T)\otimes I_d)\cdot(1_n\otimes I_d)\bartheta^k\notag\\
&=((A^k1_n-1_n{(\pi^k)}^T1_n)\otimes I_d)\bartheta^k\notag\\
&=((1_n-1_n)\otimes I_d)\bartheta^k=0.\notag
\end{align}
Then \myref{equ:l8p3} is equivalent to
\begin{align}
\tiltheta^{k+1}=M^k\tiltheta^k+\tilde{\epsilon}^k,\notag
\end{align}
where $M^k=((A^k-1_n{(\pi^k)}^T)\otimes I_d)$ and $\tilde{\epsilon}^k=\epsilon^k-{\bar{\epsilon}}^k\rightarrow0$. By using mathematical induction, it is easy to verify that
$$
\prod_{l=k}^{s-1}M^l=(A^{s:k}-1_n(\pi^k)^T)\otimes I_d.
$$
It follows from Lemma \ref{thm:ergodic}(b) that there exists $L>0$ and $0\leq\xi<1$ such that for any $i,j\in\mathcal{V}$,
$$
|a^{s:k}_{ij}-1_n(\pi^k)^T_{ij}|\leq L\xi^{s-k},
$$
and further we have
$$
\|\prod_{l=k}^{s-1}M^l\|_{\infty}\leq ndL\xi^{s-k}.
$$
In light of Lemma \ref{lemma:decay}, we obtain $\lim_{k\rightarrow\infty}\|\tiltheta^k\|_{\infty}=0$. That is, for any $j\in\mathcal{V}$, $\lim_{k\rightarrow\infty}\|\theta_j^k-\bartheta^k\|=0$.
\end{proof}
The rest of two parts both crucially depend on a property of iterative ``projections'', which implies that the distance from a point to the intersection of convex sets decreases no matter how many times the ``projections'' are performed.
\begin{mylem}[Iterative projection]\label{lemma:toset}Let $\{h_k\}:~\mathbb{R}^m\rightarrow\mathbb{R}$ be a sequence of convex functions and $\{\Omega_k\}\subseteq\mathbb{R}^m$ be a sequence of convex closed sets. Define $\{y_k\}\subseteq \mathbb{R}^m$ by
\begin{equation*}
y_{k+1}=\Pi_{\Omega_k}(y_k-\beta\frac{h_{k}(y_k)_+}{\| d_k\|^2}d_k),
\end{equation*}
where $0<\beta<2,d_k\in\partial h_{k}(y_k)$ if $h_{k}(y_k)>0$ and $d_k= d$ for any $d\neq 0$, otherwise. For any $z\in{(\Omega_0\cap\cdots\cap\Omega_{k-1})}\bigcap\{y|h_j(y)\leq0,j=0,\ldots,k-1\}$, it holds
\begin{equation*}
\| y_k-z\|^2\leq\| y_0-z\|^2-\beta(2-\beta)\frac{\| h_{0}(y_0)_+\|^2}{\| d_0\|^2}.
\end{equation*}
\end{mylem}
\begin{proof} By \cite[Lemma 1]{nedic2011random} and the definition of $\{y_k\}$, it holds for $j\leq k-1$ that $\| y_{j+1}-z\|^2\leq\| y_{j}-z\|^2-\beta(2-\beta)\frac{\| h_{j}(y_j)_+\|^2}{\| d_{j}\|^2}$. Together with the fact that $0<\beta<2$, we have that $\| y_{j+1}-z\|^2\leq\| y_{j}-z\|^2$. Then, $\| y_{k}-z\|^2\leq\| y_{1}-z\|^2\leq\| y_0-z\|^2-\beta(2-\beta)\frac{\| h_{0}(y_0)_+\|^2}{\| d_0\|^2}$.
\end{proof}
The second result essentially ensures the local feasibility.
\begin{mylem}[Feasibility]\label{lemma:feasi} The sequence $\{\theta_j^k\}$ is generated by algorithm given in \myref{equ:generic}. We define $\lambda_j^k$ and $\mu_j^k$ as follows
\begin{equation}
\label{equ:l4def}
\lambda_j^k=\sum_{i=1}^n a_{ji}^k\theta_i^k,~\text{and}~\mu_j^k=\Pi_{\Theta_0}(\lambda_j^k),
\end{equation}
where $\Theta_0$ is defined in \myref{def:theta}. For any $j\in\mathcal{V}$, if $\lim_{k\rightarrow\infty}\| \lambda_j^k-\mu_j^k\|=0$, then $\lim_{k\rightarrow\infty}\| \mu_j^k-\theta_j^{k+1}\|=0$.
\end{mylem}
\begin{proof}By Lemma \ref{lemma:toset}, we set $y_0=p_j^k$, where $p_j^k$ is given in \myref{equ:a2p1}, $h_0(y)=g_j^{\omega_j^k}(y)$ and $h_1(y)=f_j(y)$, $\Omega_0=\mathbb{R}^m$ and $\Omega_1=\Theta$. Then it follows from the algorithm given in \myref{equ:generic} that $y_2=\theta_j^{k+1}$. Since $\mu_j^k\in\Theta_0\subseteq(\Omega_0\cap\Omega_1)$, both $y_0(\mu_j^k)\leq0$ and $y_1(\mu_j^k)\leq0$ are satisfied. By Lemma \ref{lemma:toset}, it holds that
\begin{equation*}
\|\theta_j^{k+1}-\mu_j^k\|^2\leq\| p_j^k-\mu_j^k\|^2-\beta(2-\beta)\frac{g_j^{\omega_j^k}(p_j^k)_+^2}{\| d_j^k\|^2}.
\end{equation*}
Notice that $\| p_j^k-\mu_j^k\|\leq\| p_j^k-\lambda_j^k\|+\|\lambda_j^k-\mu_j^k\|=\zeta^k\|c\|+\|\lambda_j^k-\mu_j^k\|$,  we have \begin{equation}\label{proof:toset}\|\theta_j^{k+1}-\mu_j^k\|\leq\zeta^k\|c\|+\|\lambda_j^k-\mu_j^k\|.\end{equation}
If $\lim_{k\rightarrow\infty}\| \lambda_j^k-\mu_j^k\|=0$, noting that \myref{equ:stepdef} implies $\zeta^k\rightarrow0$, it holds that $\lim_{k\rightarrow\infty}\|\theta_j^{k+1}-\mu_j^k\|=0$ by taking limits on both sides of \myref{proof:toset}, which indicates that $\theta_j^{k+1}$ is infinitely approaching a feasible solution $\mu_j^k$.
\end{proof}
Finally, the last part is a stochastically ``decreasing'' result, whose proof is similar to that of \cite[Lemma 4]{you2016scenario}, and the detail is omitted here.
\begin{mylem}[Stochastically decreasing]\label{lemma:stodec}  Let $\mathcal{F}^k$ be a $\sigma$-field generated by the random variable $\{\omega_j^k,j\in\mathcal{V}\}$ up to time $k$. Under Assumptions \ref{assum:random} and \ref{assum:solvable}, it holds almost surely that for $\forall j\in\mathcal{V}$ and sufficiently large number $k$,
\begin{align}
\mathbb{E}&\left[\|\theta_j^{k+1}-\theta^*\|^2\right|\mathcal{F}^k]\notag\\
\leq& (1+R_1{(\zeta^k)}^2)\|\lambda_j^k-\theta^*\|^2-2\zeta^kc^T(\mu_j^k-\theta^*)\notag\\
\label{equ:l5res}
& -R_2\|\lambda_j^k-\mu_j^k\|^2+R_3{(\zeta^k)}^2.
\end{align}
where $\lambda_j^k$, $\mu_j^k$ are given in \myref{equ:l4def}, $\theta^*\in\Theta^*$, and $R_1, R_2, R_3$ are positive constants.
\end{mylem}

The proof also relies crucially on the well-known super-martingale convergence Theorem, which is due to \cite{robbins1985convergence}, see also \cite[Proposition A.4.5]{bertsekas2015convex}. This result is now restated for completeness and after which all previous results are summarized in a proposition.
\begin{mythm}[Super-martingale convergence]\label{thm:superm} Let $\{v_k\}$, $\{u_k\}$, $\{a_k\}$ and $\{b_k\}$ be sequences of nonnegative random variables such that
\begin{equation}
\mathbb{E}\left[v_{k+1}|\mathcal{F}_k\right]\leq (1+a_k)v_k-u_k+b_k
\end{equation}
where $\mathcal{F}_k$ denotes the collection $v_0,\dots,v_k$, $u_0,\dots,u_k$, $a_0,\dots,a_k$, $b_0,\dots,b_k$. Let $\sum_{k=0}^\infty a_k<\infty$ and $\sum_{k=0}^\infty b_k<\infty$ almost surely. Then, we have $\lim_{k\rightarrow\infty}v_k=v$ for a random variable $v\geq0$ and $\sum_{k=0}^\infty u_k<\infty$ almost surely.
\end{mythm}

\begin{myprop}[Convergent results]\label{prop}Under Assumption \ref{assum:random}-\ref{assum:strength} and let $\tillambda^k=\sum_{j=1}^n\pi_j^{k+1}\mu_j^k$, $\tilmu^k=\sum_{j=1}^n\pi_j^{k+1}\mu_j^k$, and $\bartheta^k=\sum_{j=1}^n\pi_j^k\theta_j^k$, where $\lambda_j^k$ and $\mu_j^k$ are given in \myref{equ:l4def}. Then, for any $\theta^*\in\Theta^*$ and $j\in\mathcal{V}$, the following statements hold in the almost sure sense,
\begin{enumerate}\renewcommand{\labelenumi}{\rm(\alph{enumi})}
\item $\{\sum_{j=1}^n\pi_j^k\|\theta_j^k-\theta^*\|^2\}$ converges as $k\rightarrow\infty$.
\item $\lim\inf_{k\rightarrow\infty}c^T\tilmu^k=c^T\theta^*$.
\item $\lim_{k\rightarrow\infty}\|\mu_j^k-\lambda_j^k\|=0$.
\item $\lim_{k\rightarrow\infty}\|\mu_j^k-\theta_j^{k+1}\|=\lim_{k\rightarrow\infty}\|\lambda_j^k-\theta_j^{k+1}\|=0$.
\item $\lim_{k\rightarrow\infty}\|\tilmu^k-\bartheta^{k+1}\|=\lim_{k\rightarrow\infty}\|\tillambda^k-\bartheta^{k+1}\|=0$.
\end{enumerate}
\end{myprop}
\begin{proof} By the convexity of $\|\cdot\|^2$ and the row stochasticity of $A^k$, i.e, $\sum_{i=1}^na_{ji}^k=1$, it follows that
\begin{equation*}
\|\lambda_j^k-\theta^*\|^2\leq\sum_{i=1}^na_{ji}^k\|\theta_i^k-\theta^*\|^2.
\end{equation*}
Jointly with \myref{equ:l5res}, we obtain that  for sufficiently large $k$,
\begin{align}
\mathbb{E}&\left[\|\theta_j^{k+1}-\theta^*\|^2\mathcal{F}^k\right]\notag\\ \leq&(1+R_1{(\zeta^k)}^2)\sum_{i=1}^na_{ji}^k\|\theta_i^k-\theta^*\|^2-2\zeta^kc^T(\mu_j^k-\theta^*)\notag\\
\label{equ:p1pr1}
& -R_2\|\lambda_j^k-\mu_j^k\|^2+R_3{(\zeta^k)}^2.
\end{align}
Premultiply $\pi^{k+1}_j$ on both sides of \myref{equ:p1pr1} and sum up on $j$,
\begin{align}
\mathbb{E}&\big[\sum_{j=1}^n\pi_j^{k+1}\|\theta_j^{k+1}-\theta^*\|^2|\mathcal{F}^k\big]\notag\\
\leq&(1+R_1{(\zeta^k)}^2)\sum_{j=1}^n\sum_{i=1}^n\pi_j^{k+1}a_{ji}^k\|\theta_i^k-\theta^*\|^2-2\zeta^kc^T\notag\\
&\cdot(\tilmu^k-\theta^*)-R_2\sum_{j=1}^n\pi_j^{k+1}\|\lambda_j^k-\mu_j^k\|^2+R_3{(\zeta^k)}^2\notag\\
=&(1+R_1{(\zeta^k)}^2)\sum_{j=1}^n\pi_j^k\|\theta_j^k-\theta^*\|^2-2\zeta^kc^T(\tilmu^k-\theta^*)\notag\\
\label{equ:p1pr2}
&-R_2\sum_{j=1}^n\pi_j^{k+1}\|\lambda_j^k-\mu_j^k\|^2+R_3{(\zeta^k)}^2,
\end{align}
the last equality holds due to Lemma \ref{thm:ergodic}(b). It follows from \myref{equ:stepdef} that $\sum_{k=0}^\infty R_1{(\zeta^k)}^2<\infty$ and $\sum_{k=0}^\infty R_3{(\zeta^k)}^2<\infty$. Notice the convexity of $\Theta_0$ and  $\mu_j^k\in\Theta_0$, it is clear that $\tilmu^k\in\Theta_0$. In view of the fact that $\theta^*$ is one optimal solution in $\Theta_0$, it holds that $c^T\tilmu^k-c^T\theta^*\geq 0$. Thus, all the conditions in Theorem \ref{thm:superm} are satisfied. Therefore, it holds almost surely that $\{\sum_{j=1}^n\pi_j^k\|\theta_j^k-\theta^*\|^2\}$ converges for any $j\in\mathcal{V}$ and $\theta^*\in\Theta^*$. Hence, Proposition \ref{prop}(a) is proved. Moreover, it follows from Theorem \ref{thm:superm} that
\begin{equation}
\label{equ:l7rb}
\sum_{k=0}^\infty\zeta^kc^T(\tilmu^k-\theta^*)<\infty
\end{equation}
and
\begin{equation}
\label{equ:l7rc}
\sum_{k=0}^\infty\sum_{j=1}^n\pi_j^{k+1}\|\lambda_j^k-\mu_j^k\|^2<\infty.
\end{equation}
It is clear that \myref{equ:l7rb} directly implies Proposition \ref{prop}(b) under the condition $c^T\tilmu^k-c^T\theta^*\geq0$. We know that $\pi_i^k\geq\eta>0$ from Lemma \ref{thm:ergodic}(c), then \myref{equ:l7rc} directly shows that $\lim_{k\rightarrow\infty}\|\lambda_j^k-\mu_j^k\|^2=0$ for any $j\in\mathcal{V}$. Thus, Proposition \ref{prop}(c) is proved. Combining the result in Proposition \ref{prop}(c) with Lemma \ref{lemma:feasi}, it is clear that Proposition \ref{prop}(d) holds as well. As for Proposition \ref{prop}(e), we notice that
\begin{align}
\|\tilmu^k-\bartheta^{k+1}\|&=\|\sum_{j=1}^n\pi_j^{k+1}\mu_j^k-\sum_{j=1}^n\pi_j^{k+1}\theta^{k+1}\|\notag\\
&\leq\sum_{j=1}^n\pi_j^{k+1}\|\mu_j^k-\theta_j^{k+1}\|.\notag
\end{align}
By taking limits on both sides and using the result in Proposition 2(d), we have $\lim_{k\rightarrow\infty}\|\tilmu^k-\bartheta^{k+1}\|=0$, and it similarly holds that $\lim_{k\rightarrow\infty}\|\tillambda^k-\bartheta^{k+1}\|=0$.
\end{proof}

Now, we are in a position to prove Theorem \ref{thm:convergent}.

{\em Proof of Theorem \ref{thm:convergent}.} Notice that $\lambda_j^k=\sum_{i=1}^na_{ji}^k\theta_j^k$, it follows from Proposition \ref{prop}(d) that $\lim_{k\rightarrow\infty}\|\theta_j^{k+1}-\sum_{i=1}^na_{ji}^k\theta_j^k\|=0$. Then it holds from Lemma \ref{lemma:consensus}  almost surely that $\lim_{k\rightarrow\infty}\|\theta_j^k-\bartheta^k\|=0$. Together with Proposition \ref{prop}(a) and the row-stochasticity of $A^k$, we obtain that $\{\|\bartheta^k-\theta^*\|\}$ converges. We know from Proposition \ref{prop}(e) that $\tilmu^k\rightarrow\bartheta^{k-1}$ as $k\rightarrow\infty$, so $\{\|\tilmu^k-\theta^*\|\}$ converges as well. Consider Proposition \ref{prop}(b), it implies that there exists a subsequence $\{\tilmu^k|k\in\mathcal{K}\}$ that converges almost surely to some point in the optimal set $\Theta^*$, which is denoted as $\theta_0^*$, and it holds clearly that
$$
\lim_{k\in\mathcal{K},k\rightarrow\infty}\|\tilmu^k-\theta_0^*\|=0.
$$
Since $\{\|\tilmu^k-\theta_0^*\|\}$ converges, it follows that $\lim_{k\rightarrow\infty}\|\tilmu^k-\theta_0^*\|=0$. Finally, we note that $\|\theta_j^{k+1}-\theta_0^*\|\leq\|\theta_j^{k+1}-\bartheta^{k+1}\|+\|\bartheta^{k+1}-\tilmu^k\|+\|\tilmu^k-\theta_0^*\|$, which converges almost surely to zero as $k\rightarrow\infty$. Therefore, there exists $\theta_0^*\in\Theta^*$ such that $\lim_{k\rightarrow\infty}\theta_j^k=\theta_0^*$ for all $j\in\mathcal{V}$ with probability one. Thus, Theorem \ref{thm:convergent} is proved.

\section{Conclusion}
\label{Conclusion}
In this work, we developed a simple structural algorithm  to collaboratively solve distributed constrained optimization problems over time-varying unbalanced digraphs. Convergence results are rigorously shown using stochastic theory. The main drawback of the proposed algorithm is that the number of the augmented variables depends on the scale of topology in agreement with previous research on the topic. Future work will focus on reducing the number of augmented variables and speeding up the convergent rate.

\bibliographystyle{IEEEtrans}
\bibliography{mybibf}
\end{document}